\documentclass[conference]{IEEEtran}
\IEEEoverridecommandlockouts

\usepackage{cite}
\usepackage{amsmath,amssymb,amsfonts, amsthm}
\usepackage{algorithmic}
\usepackage{graphicx}
\usepackage[caption=false]{subfig}
\usepackage{textcomp}
\usepackage{xcolor}
\usepackage{booktabs}
\usepackage{braket}
\usepackage{mathtools}
\usepackage{hyperref}
\usepackage{cleveref}
\usepackage{pgf}
\usepackage{pgfplots}
\pgfplotsset{compat=1.18}

\usepackage{subcaption}
\crefname{section}{Sec.}{Secs.}
\Crefname{section}{Sec.}{Secs.}
\crefname{subsection}{Sec.}{Secs.}
\Crefname{subsection}{Sec.}{Secs.}
\crefname{subsubsection}{Sec.}{Secs.}
\Crefname{subsubsection}{Sec.}{Secs.}
\Crefname{figure}{Fig.}{Figs.}
\Crefname{table}{Tab.}{Tabs.}
\usepackage{quantikz}
\usepackage{orcidlink}
\DeclareMathOperator{\Tr}{tr}

\usepackage{stmaryrd}
\def\BibTeX{{\rm B\kern-.05em{\sc i\kern-.025em b}\kern-.08em
    T\kern-.1667em\lower.7ex\hbox{E}\kern-.125emX}}
\begin{document}

\title{Branch-Aware Quantum Constant Propagation for Dynamic Quantum Circuits\\
}

\author{
\IEEEauthorblockN{Innocenzo Fulginiti \orcidlink{0000-0001-8818-9626}}
\IEEEauthorblockA{\textit{TUM School of CIT} \\
\textit{Technical University of Munich}\\
Garching, Germany \\
innocenzo.fulginiti@tum.de}
\and
\IEEEauthorblockN{Yanbin Chen \orcidlink{0000-0002-1123-1432}}
\IEEEauthorblockA{\textit{TUM School of CIT} \\
\textit{Technical University of Munich}\\
Garching, Germany \\
yanbin.chen@tum.de }
}

\newcommand{\topq}{\top_{\mathfrak{q}}}
\newcommand{\topc}{\top_{\mathfrak{c}}}
\newcommand{\atrans}[2]{\mathfrak{T}^{\sharp}(#1)\left(#2\right)}
\newcommand{\qtrans}[2]{\mathsf{T}^{\sharp,Q}_{#1}\!\left(#2\right)}
\newcommand{\ctrans}[2]{\mathfrak{T}(#1)\!\left(#2\right)}
\newcommand{\qtransc}[2]{\mathsf{T}^{Q}_{#1}\!\left(#2\right)}

\newtheorem{remark}{Remark}
\newtheorem{lemma}{Lemma}
\newtheorem{theorem}{Theorem}

\maketitle

\begin{abstract}
Compile-time optimization is important for improving the efficiency and reliability of quantum circuits on current noisy hardware. While many existing methods simplify circuits using structural patterns or quantum-state information, most of them target only unitary circuits and do not support dynamic circuits with mid-circuit measurements and classical feedforward. In this work, we present Branch-Aware Quantum Constant Propagation (BQCP), a compile-time analysis for dynamic circuits. BQCP extends Quantum Constant Propagation (QCP) by tracking the classical information produced by mid-circuit measurements together with the corresponding post-measurement quantum states across different execution branches. This enables path-sensitive reasoning inside conditional blocks and more precise information propagation than QCP. To keep the analysis scalable, we bound both the size of the quantum-state representation and the number of tracked branches. Using the information inferred by the analysis, we apply semantics-preserving simplifications to circuit operations. We prove the soundness of both the analysis and the simplifications. Experimental results on both application-driven and synthetic benchmarks show that, on dynamic circuits, our method consistently achieves larger reductions than other existing passes including QCP.
\end{abstract}


\section{Introduction}
\label{sec:intro}
Current quantum hardware is highly constrained by noise and decoherence \cite{Preskill2018quantumcomputingin, RevModPhys.94.015004}. Such limitations make it essential to apply optimizations and simplifications to quantum circuits before execution, as even small inefficiencies could worsen the computation's reliability. In this context, compile-time optimization passes play a key role, enabling semantic-preserving simplifications that reduce circuit complexity and improve resource usage during execution.
With the growing availability of quantum programming frameworks and libraries, programmers increasingly rely on high-level language constructs and automated circuit-generation pipelines, including synthesis and general-purpose program transformations, to produce executable circuits from given input specifications.
While these approaches improve abstraction and simplify programmability, they are designed to remain correct across a wide range of inputs and therefore may not produce circuits that are fully optimized for each particular instance. 
As a consequence, there is room for compile-time analyses to identify simplifications that are not apparent from the high-level description.
Over the years, several compile-time optimization techniques have been proposed to simplify quantum circuits, including both rewriting-based techniques and analyses that track properties of the quantum state \cite{assertion-based, rpo, qcp}.
Modern quantum programming languages and SDKs \cite{openqasm3, ibmClassicalFeedforward, quantumaiClassicalControl} support dynamic quantum circuits, where unitary operations are interleaved with mid-circuit measurements, resets, and classical feedforward. Dynamic circuit features are exploited in several application domains, including quantum error correction protocols, qubit reuse techniques, and circuit cutting methods~\cite{exploiting_dyn_circ, qec_1, qubit_reuse_1, qubit_reuse_2, gate_cutting_alg_qdislib}.
However, most existing compile-time optimization passes have limited effectiveness on dynamic circuits, since they are not designed to reason about the dynamic control flow induced by mid-circuit measurements and conditionals \cite{assertion-based, rpo, qcp}.

In this work, we propose a compile-time analysis for dynamic quantum circuits that tracks the classical information generated by mid-circuit measurements and uses it to reason about subsequent control flow. This enables reasoning across the different branches that may arise during circuit execution. Our approach builds on Quantum Constant Propagation (QCP)~\cite{qcp}, a compilation pass that propagates quantum information from a fixed initial state and exploits it to simplify controlled gates by removing redundant controls or eliminating gates whose control conditions are unsatisfiable.
Previous works have also used QCP to rewrite circuits containing mid-circuit measurements and resets into equivalent \emph{probabilistic circuits} \cite{prob, big_prob, fulginiti2026}. However, since they employ original QCP, those approaches do not reason about the different execution branches induced by dynamic circuits.
We introduce \emph{Branch-Aware Quantum Constant Propagation} (BQCP), an extension of QCP for dynamic circuits.
By preserving the correlation between classical measurement results and the corresponding post-measurement quantum states across execution branches, BQCP enables path-sensitive reasoning inside conditional blocks.
Propagating quantum information is inherently challenging because entanglement can induce an exponential growth in the dimension of the quantum state space \cite{PhysRevX.10.041038}.
QCP mitigates this issue by tracking quantum states only up to a fixed size bound, and conservatively stopping the tracking of states that exceeded the bound. In BQCP, an additional source of complexity is the growth in the number of execution paths, which we control through a bounded abstraction that limits both the tracked quantum information and the number of tracked branches, while preserving soundness.
This gives a tunable trade-off between precision and cost that can be adapted to the structure of the analyzed circuits.
The information tracked by our branch-aware analysis allows us to identify optimizations that are not captured by existing approaches \cite{assertion-based, qcp, rpo}. Beyond simplifying controls as in QCP, we exploit propagated classical--quantum information to remove both unitary and non-unitary operations that are redundant on the current abstract state.
Besides directly shrinking circuits, these transformations also simplify subsequent compilation stages, as removing unnecessary multi-qubit gates makes qubit mapping, routing, and decomposition into the target gate set easier~\cite{routing_1, routing_2, Rosa2025optimizinggate, rosa2025quantumgatedecompositionstudy}, resulting in more compact compiled circuits.
We formalize our analysis and prove that it soundly over-approximates the concrete classical--quantum semantics of dynamic circuits, and that all simplifications applied exploiting propagated information preserve program semantics.
We implement and evaluate BQCP on both application-driven and synthetic benchmarks. On a circuit cutting case study, we show that branch-aware propagation exposes optimization opportunities in dynamic circuits that are missed by other optimization passes. To complement this use case, we also compare BQCP against original QCP on a broader and more diverse set of randomly generated dynamic circuits. This second benchmark provides a more systematic assessment of the optimization gains enabled by branch-aware analysis, showing that BQCP consistently achieves larger reductions than QCP.

The implementation of the proposed pass is publicly available at
\url{https://github.com/1nnocenzo/bqcp}.

\section{Preliminaries}
\label{sec:prelim}
This section reviews the background concepts used in the remainder of the paper. We briefly recall QCP, introduce dynamic circuits and their classical–quantum semantics, and present a motivating example illustrating the optimization opportunities enabled by a branch-aware analysis.
\subsection{Quantum Constant Propagation (QCP)}
\label{subsec:prelim-qcp}
QCP is a static analysis method that propagates abstract information about the quantum state along a circuit. Starting from an initial state, it traverses the circuit and updates the abstract state to reflect the effect of each gate \cite{qcp}.
\subsubsection{Quantum abstract state}
\label{prelim:subsec-qas}
Consider a circuit over a quantum register $Q = \{q_0,\dots,q_{n-1}\}$. QCP maintains a partition $\mathcal{G}$ of $Q$ into \textit{entanglement groups} such that, at each program point and for any $q_i,q_j \in Q$,
$q_i \text{ and } q_j \text{ are entangled}$ only if there exists $G \in \mathcal{G}$ such that $q_i,q_j \in G$.
In this paper, we use the refined variant of QCP described in \cite{big_prob}, in which qubits belong to the same group if and only if they are entangled. Accordingly, each group $G \in \mathcal{G}$ is tracked independently: entangling operations merge previously separate groups, whereas disentangling operations refine $\mathcal{G}$ by splitting a group into smaller ones.
Fix, for each group $G$, a canonical ordering of its qubits, so that basis strings range over $\{0,1\}^{|G|}$.
Let the concrete state of $G$ be the pure state
\[
  \ket{\psi_G} = \sum_{x\in\{0,1\}^{|G|}} \alpha_x \ket{x}, \quad \text{with } \sum_x |\alpha_x|^2 = 1.
\]
QCP represents $\ket{\psi_G}$ sparsely as a mapping from basis strings $x$ with $\alpha_x\neq 0$ to their amplitudes $\alpha_x\in\mathbb{C}$.
To ensure polynomial-time analysis and mitigate the exponential growth induced by entanglement,
QCP fixes a threshold $n_{\max}$ on the number of non-zero basis states that can be represented per group.
As long as the support size of $\ket{\psi_G}$ is at most $n_{\max}$, this representation is exact (i.e., it contains all non-zero amplitudes).
If applying an instruction would increase the number of non-zero basis states beyond $n_{\max}$, the group is marked as untracked.
For a group $G\in\mathcal{G}$, let
\[
\mathcal{Q}(G) \triangleq \{\, s:\{0,1\}^{|G|}\rightharpoonup \mathbb{C} \; s.t.\; |\operatorname{dom}(s)| \le n_{\max} \,\} \ \cup\ \{\topq\},
\]
where $\topq$ denotes an untracked group, and $s$ is a finite partial function mapping computational-basis strings to amplitudes. As long as $s\neq \topq$, $s$ stores all and only the non-zero amplitudes of $\ket{\psi_G}$, with $\alpha_x = s(x)$ for $x\in\operatorname{dom}(s)$, and $\alpha_x = 0$ otherwise.
Given a partition $\mathcal{G}$ of the quantum register $Q$, the corresponding abstract domain is defined as
\[
\mathcal{Q}(\mathcal{G}) \triangleq \prod_{G\in\mathcal{G}} \mathcal{Q}(G).
\]
An abstract quantum state is then any element $\gamma_Q \in \mathcal{Q}(\mathcal{G})$.
Since the abstract quantum state \(\gamma_Q\) depends on the partition \(\mathcal{G}\), the abstract domain for \(Q\) is defined as
\[
\mathcal{Q}_Q \triangleq \{\,(\mathcal{G},\gamma_Q) \mid \gamma_Q \in \mathcal{Q}(\mathcal{G})\,\}.
\]
In the remainder of the paper, for \(\gamma_Q \in \mathcal{Q}(\mathcal{G})\) and \(G \in \mathcal{G}\), we write \(\gamma_Q(G)\) for the abstract state associated with \(G\).

\subsubsection{Control reduction}
\label{subsec:prelim-qcp-cont-red}
QCP exploits the current abstract quantum state to simplify controls in controlled gates (e.g., $CZ$, $CCX$). Sufficient conditions include: a control qubit is deterministically in $\ket{0}$, in which case the gate never triggers; and a control qubit is deterministically in $\ket{1}$, in which case the control can be dropped. QCP also detects when the conjunction of control conditions is unsatisfiable, in which case the gate can never be enabled and can be removed.

\subsection{Dynamic quantum circuits}
\label{subsec:prelim-dqc}
A dynamic quantum circuit over an $n$-qubit quantum register $Q=\{q_0,\dots,q_{n-1}\}$ and an $m$-bit classical register
$C=\{c_0,\dots,c_{m-1}\}$ is a sequence of instructions, including, in addition to unitary gates, mid-circuit measurements that write to the classical register, resets, and conditionals, namely \textit{if-then-else} constructs guarded by classical conditions, i.e., predicates over the classical register \cite{openqasm3, ibmClassicalFeedforward, quantumaiClassicalControl}.
\subsubsection*{Concrete classical--quantum semantics}
\label{subsec:pre-qc-sem}
Unitary gates act deterministically on the quantum register, while non-unitary operations such as measurements may produce different outcomes, thereby giving rise to different execution branches. Each branch is identified by a classical register configuration together with the corresponding post-measurement quantum state. The concrete classical--quantum state at a given program point is therefore a finite set of execution branches.
Let $\mathcal{H}_Q = (\mathbb{C}^2)^{\otimes n}$ denote the Hilbert space associated with the $n$-qubit register $Q$, and let $\mathcal{D}(\mathcal{H}_Q)$ denote the set of density operators over $\mathcal{H}_Q$. Let $\{0,1\}^m$ be the set of configurations of the $m$-bit register $C$. The concrete classical--quantum semantics $\mathfrak{S}$ of a dynamic circuit over $Q$ and $C$ is a set
\[
  \mathfrak{S} \;\subseteq\;
  \{0,1\}^m \times \mathcal{D}(\mathcal{H}_Q).
\]
An element $(c,\rho) \in \mathfrak{S}$ represents a reachable execution branch in which $C$ has value $c$ and $Q$ is in state $\rho$.
For the purposes of our static analysis, only the reachability of branches matters, so we ignore their probabilities.

\subsection{Motivating example}
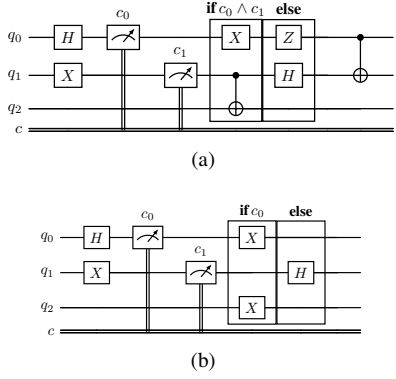
\begin{figure}[t]
\centering
\subfloat[\label{fig:two_circuits:orig}]{%
\resizebox{0.6\linewidth}{!}{%
  \begin{quantikz}[row sep=0.3cm, column sep=0.55cm]
\lstick{$q_0$} & \gate{H} & \meter{c_0}\wire[d][3]{c} &
& \gate{X}\gategroup[wires=3,steps=1,
  style={draw,inner xsep=4pt,inner ysep=0.1pt},
  label style={label position=above}]{\textbf{if}$\,c_0 \land c_1$}
& \gate{Z} \gategroup[wires=3,steps=1,
  style={draw,inner xsep=4pt,inner ysep=0pt, outer ysep=0pt},
  label style={label position=above}]{\textbf{else}}
& \qw & \ctrl{1} &
\\
\lstick{$q_1$} & \gate{X} & \qw & \meter{c_1}\wire[d][2]{c} & \ctrl{1}
& \gate{H}
& \qw & \targ{} &
\\
\lstick{$q_2$} & \qw & \qw & \qw & \targ{} & \qw & \qw & \qw &
\\
\lstick{$c$}   & \setwiretype{c} \cw & \cw & \cw & \cw & \cw & \cw & \cw & \cw
\end{quantikz}
}} \\
\subfloat[\label{fig:two_circuits:simpl}]{%
\resizebox{0.5\linewidth}{!}{%
  \begin{quantikz}[row sep=0.3cm, column sep=0.55cm]
\lstick{$q_0$} & \gate{H} & \meter{c_0}\wire[d][3]{c} &
& \gate{X}\gategroup[wires=3,steps=1,
  style={draw,inner xsep=4pt,inner ysep=0.1pt},
  label style={label position=above}]{\textbf{if}$\,c_0$}
& \qw \gategroup[wires=3,steps=1,
  style={draw,inner xsep=4pt,inner ysep=0pt, outer ysep=0pt},
  label style={label position=above}]{\textbf{else}}
& \qw &
\\
\lstick{$q_1$} & \gate{X} & \qw & \meter{c_1}\wire[d][2]{c} & \qw
& \gate{H}
& \qw &
\\
\lstick{$q_2$} & \qw & \qw & \qw & \gate{X} & \qw & \qw &
\\
\lstick{$c$}   & \setwiretype{c} \cw & \cw & \cw & \cw & \cw & \cw & \cw
\end{quantikz}
}}
\caption{Dynamic circuit example assuming all qubits are initially in $\ket{0}$ and all classical bits are $0$: Fig.~\ref{fig:two_circuits:orig} shows the original circuit, and Fig.~\ref{fig:two_circuits:simpl} shows a semantically equivalent simplified circuit.}
\label{fig:opt_exmp}
\end{figure}
Existing optimization passes such as QCP are designed for circuits with a single linear execution branch and do not reason in terms of the concrete classical--quantum semantics of dynamic circuits illustrated in \Cref{subsec:pre-qc-sem}. In the presence of non-unitary operations and conditionals, they lose information about the state of the involved qubits, thereby limiting the optimization opportunities that can be detected.
Fig.~\ref{fig:opt_exmp} illustrates this situation on a simple dynamic circuit example, assuming the initial configuration in which all qubits are initialized to $\ket{0}$ and all classical bits to $0$. By propagating classical--quantum information and tracking the correlation between the measurement outcomes and the post-measurement quantum state along each execution branch, it is possible to detect that: upon entering the \textit{then} block, the state of $q_1$ is $\ket{1}$, allowing to replace the $CX$ gate inside the \textit{then} block with a non-controlled $X$ gate; upon entering the \textit{else} branch, the state of $q_0$ is $\ket{0}$, then the $Z$ gate produces no effect on it. Moreover, after the \textit{if-then-else} operation, $q_0$ is in state $\ket{0}$ regardless of which branch is taken, therefore, the final $CX$ can be omitted.
Passes not designed for dynamic circuits are unable to detect these simplifications. This motivates the need for a branch-aware analysis pass, which we introduce in the next section.

\section{Method}
\label{sec:method}
We now present Branch-Aware Quantum Constant Propagation (BQCP). We first introduce the abstract domain used to represent classical--quantum states. We then define the abstract transfer functions for the circuit instructions, which update the abstract state. Finally, we show how the propagated information is exploited to simplify dynamic circuits through semantics-preserving optimizations.
\subsection{Abstract state and domain}
\label{sec:method-domains}
We now introduce an abstract semantics for dynamic circuits, designed to model the concrete classical--quantum semantics for dynamic circuits (\Cref{subsec:pre-qc-sem}).
\subsubsection{Classical abstract state}
\label{sec:method-classical-domain}
Consider an $m$-bit classical register $C$. We abstract each bit $c_i \in C$ using the flat lattice
\(
  \mathcal{C} \triangleq \{0,1,\topc\},
\)
where $0$ and $1$ denote known constants, and $\topc$ denotes an unknown value. The partial order $\sqsubseteq_{\mathcal{C}}$ is defined by
\(
  0 \sqsubseteq_{\mathcal{C}} \topc, \,
  1 \sqsubseteq_{\mathcal{C}} \topc,
\)
with $0$ and $1$ incomparable. 
The abstract domain of the $m$-bit classical register $C$ is
\[
  \mathcal{C}_C \triangleq \prod_{i=0}^{m-1} \mathcal{C}.
\]
An abstract classical state is then any element $\gamma_C \in \mathcal{C}_C$.
In the remainder of the paper, we write $\gamma_C(c_i)$ to denote the abstract value of $c_i$ in the abstract state $\gamma_C$.
Given $c_i \in C$ and $v \in \mathcal{C}$, we denote by $\gamma_C[c_i \mapsto v]$ the abstract state obtained from $\gamma_C$ by updating $c_i$ to $v$ and leaving all other bits unchanged.

\subsubsection{Classical--quantum abstract state}
\label{sec:method-cq-domain}
In accordance with the concrete classical--quantum semantics of \Cref{subsec:pre-qc-sem}, we represent the abstract classical--quantum state of a dynamic circuit at a given point during execution as a finite disjunction of abstract branches.
Each branch consists of a classical component and a quantum component. The classical component is an abstract classical state \(\gamma_C \in \mathcal{C}_C\) (\Cref{prelim:subsec-qas}). For the quantum component, we use the QCP abstraction recalled in \Cref{prelim:subsec-qas}, namely, an object \(\eta_Q \in \mathcal{Q}_Q\), where \(\eta_Q \triangleq (\mathcal{G}, \gamma_Q)\). Accordingly, the abstract state of a single branch is represented as a pair
\(
(\gamma_C,\eta_Q) \in \mathcal{C}_C \times \mathcal{Q}_Q.
\)
Formally, the abstract state of a dynamic circuit is a finite disjunction of abstract branches:
\[
\mathcal{B} \triangleq \{\,(\gamma_C^1,\eta_Q^1),\dots,(\gamma_C^k,\eta_Q^k)\,\}
\ \subseteq\ \mathcal{C}_C \times \mathcal{Q}_Q.
\]
Thus, \(\mathcal{B}\) over-approximates the set of possible execution branches that may arise during the dynamic circuit execution.


\subsubsection{Bounded abstract state representation}
\label{sec:method-branch-bound}
To ensure that the analysis remains tractable, we adopt a bounded and conservative representation of the abstract classical--quantum state, with bounds imposed at two levels.
First, for the quantum component of each branch, we use the bounded QCP representation (\Cref{prelim:subsec-qas}): each entanglement group is tracked only while its number of basis states is at most $n_{\max}$, and if this bound is exceeded, the group is conservatively marked as untracked.
Second, we fix a bound $b_{\max}\in\mathbb{N}$ that controls the number of branches represented explicitly during the analysis. As long as the bound permits it, branches are tracked separately. When representing all alternatives explicitly would exceed this threshold, the analysis instead incorporates their effects conservatively through over-approximation, without introducing additional branches. This preserves soundness while allowing a tunable trade-off between precision and scalability. The precise enforcement of these bounds is defined later in the paper through the abstract transfer functions, which specify the evolution of the abstract state across circuit instructions.

\subsubsection{Conservative merge of branches}
\label{sec:method-branch-merge}
We now define a conservative merge operator for combining two abstract branches
\(
A=(\gamma_C^A,\eta_Q^A)\quad\text{and}\quad B=(\gamma_C^B,\eta_Q^B)
\)
in \(\mathcal{C}_C \times \mathcal{Q}_Q\), where
\(\gamma_C^A,\gamma_C^B\in\mathcal{C}_C\) denote their classical components, and
\(\eta_Q^A=(\mathcal{G}^A,\gamma_Q^A)\), \(\eta_Q^B=(\mathcal{G}^B,\gamma_Q^B)\in\mathcal{Q}_Q\)
their quantum components.
\paragraph{Classical component}
We merge the abstract classical state bitwise using the join operator $\sqcup_{\mathcal{C}}$ on $\mathcal{C}$. For $v,w\in\mathcal{C}$, we define $v \sqcup_{\mathcal{C}} w$ to be $v$ if $v=w$, and $\topc$ otherwise.
The merged classical state $\gamma_C^{A\sqcup B}$ is defined by
\[
  \gamma_C^{A\sqcup B}(c_i) \ \triangleq \ \gamma_C^A(c_i)\ \sqcup_{\mathcal{C}}\ \gamma_C^B(c_i),
  \quad i\in\{0,\dots,m-1\}.
\]




\paragraph{Quantum component}
To merge the quantum component, we preserve only those entanglement groups
\(G\) that appear in both partitions \(\mathcal{G}^A\) and \(\mathcal{G}^B\).
For every such group \(G\in\mathcal{G}^A\cap\mathcal{G}^B\), the merged abstract state is defined as
\[\gamma_Q^{A\sqcup B}(G) \triangleq \gamma_Q^A(G)\ \sqcup_{\mathcal{Q}}\ \gamma_Q^B(G),\]
where for \(x,y\in\mathcal{Q}(G)\) we define \(x \sqcup_{\mathcal{Q}} y \triangleq x\) if \(x=y\), and \(\topq\) otherwise.
All remaining qubits, namely those that do not belong to any preserved group, lose all tracked quantum information in the merged branch.
For representational purposes, these qubits are collected into a residual group whose abstract value is \(\topq\).
The merged partition \(\mathcal{G}^{A\sqcup B}\) consists of all entanglement groups \(G \in \mathcal{G}^A \cap \mathcal{G}^B\), possibly augmented with a residual group collecting all remaining qubits.
The merged quantum component is
\(\eta_Q^{A\sqcup B}\triangleq(\mathcal{G}^{A\sqcup B},\gamma_Q^{A\sqcup B})\).

\paragraph*{Branch merge operator} The operator for merging two abstract branches is defined as
\[A \sqcup B \triangleq (\gamma_C^{A\sqcup B},\eta_Q^{A\sqcup B}).\]
\subsection{Evaluation of classical guards}
\label{sec:method-guards}
Dynamic circuits include guarded constructs, where the guard is a predicate over the classical register \(C\). During propagation, the guard is evaluated on the abstract classical state \(\gamma_C\) to determine how control flow may proceed.

\subsubsection{Guard definition}
\label{sec:method-guards-lang}
We consider guards $\varphi$ defined as
\[
\varphi \;::=\; c_i \mid \neg\varphi \mid \varphi \land \varphi \mid \varphi \lor \varphi,
\]
where $c_i\in C$ is a bit, and $\land, \lor, \neg$ are boolean operators.

\subsubsection{Guard abstract evaluation}
\label{sec:method-guards-eval}
We evaluate guards using a three-valued logic:
\(
\mathbb{B}_3 \triangleq \{\mathsf{t},\mathsf{f},\mathsf{u}\},
\)
where \(\mathsf{t}\), \(\mathsf{f}\), and \(\mathsf{u}\) denote definitely true, definitely false, and unknown, respectively.
For each guard \(\varphi\), we define an abstract evaluation function
\[
\llbracket \varphi \rrbracket^{\sharp} : \mathcal{C}_C \to \mathbb{B}_3,
\]
which evaluates \(\varphi\) on an abstract classical state \(\gamma_C \in \mathcal{C}_C\).
For \(c_i \in C\), \(\llbracket c_i \rrbracket^{\sharp}(\gamma_C)\) is \(\mathsf{t}\) if \(\gamma_C(c_i)=1\), \(\mathsf{f}\) if \(\gamma_C(c_i)=0\), and \(\mathsf{u}\) if \(\gamma_C(c_i)=\topc\).
The boolean operators \(\neg\), \(\land\), \(\lor\) are interpreted over \(\mathbb{B}_3\) to account for unknown values in the abstract state.
For negation, we set \(\neg\mathsf{t}=\mathsf{f}\), \(\neg\mathsf{f}=\mathsf{t}\), \(\neg\mathsf{u}=\mathsf{u}\).
For conjunction, \(b_1\land b_2\) is \(\mathsf{f}\) if either operand is \(\mathsf{f}\), it is \(\mathsf{t}\) if both are \(\mathsf{t}\), and it is \(\mathsf{u}\) otherwise.
Dually, \(b_1\lor b_2\) is \(\mathsf{t}\) if either operand is \(\mathsf{t}\), it is \(\mathsf{f}\) if both are \(\mathsf{f}\), and it is \(\mathsf{u}\) otherwise.
The abstract evaluation of compound guards, namely \(\llbracket \neg\varphi \rrbracket^{\sharp}\), \(\llbracket \varphi_1\land\varphi_2 \rrbracket^{\sharp}\), and \(\llbracket \varphi_1\lor\varphi_2 \rrbracket^{\sharp}\), is then defined compositionally.

\subsection{Abstract transfer functions}
\label{sec:method-transfer}
To propagate abstract classical--quantum information through a dynamic circuit, we associate each circuit instruction \(I\) with an abstract transfer function \(\atrans{I}{\mathcal{B}}\), where \(\mathcal{B}\subseteq \mathcal{C}_C\times\mathcal{Q}_Q\) is the current abstract state. Applying \(\atrans{I}{\mathcal{B}}\) updates each branch of \(\mathcal{B}\) according to the semantics of \(I\), yielding the abstract state after executing \(I\).
Moreover, transfer functions enforce the bounds \(n_{\max}\) and \(b_{\max}\) to control the growth of the abstract state.

\subsubsection{Transfer function for unitary operations}
Let \(U\) be a unitary gate. Since unitary instructions do not act on the classical register, they preserve the classical abstraction of each branch and only update the quantum abstraction using the underlying QCP transfer function~\cite{qcp}. We denote by \(\tau^\sharp_U : \mathcal{Q}_Q \to \mathcal{Q}_Q\) the corresponding quantum transfer function.
Given \(\eta_Q=(\mathcal{G},\gamma_Q)\), \(\tau^\sharp_U(\eta_Q)\) updates the entanglement partition
\(\mathcal{G}\) and the abstract quantum state \(\gamma_Q\), enforcing the bound \(n_{\max}\) by setting to \(\topq\)
any group whose tracked support would exceed \(n_{\max}\) \cite{qcp}.
In our analysis, the transfer function for a unitary instruction \(U\) is defined over an abstract state
\(\mathcal{B}\subseteq \mathcal{C}_C\times\mathcal{Q}_Q\).
It acts branchwise by preserving the classical component of each branch and applying \(\tau^\sharp_U\) to its quantum component. Formally, we define
\[
  \atrans{U}{\mathcal{B}}
  \;\triangleq\;
  \{\, (\gamma_C,\tau^{\sharp}_{U}(\eta_Q)) \mid (\gamma_C,\eta_Q)\in\mathcal{B}\,\}.
\]

\subsubsection{Transfer function for measurements}
\label{sec:method-transfer-measure}
Let $M_{q_j \to c_i}$ be the instruction that measures the qubit $q_j\in Q$ and stores
the outcome in the classical bit $c_i\in C$.
To define the measurement transfer function, we first describe how a measurement \(M_{q_j \to c_i}\) updates a single abstract branch $(\gamma_C,(\mathcal{G},\gamma_Q))$, distinguishing deterministic from non-deterministic measurements.
\paragraph{Deterministic measurement}
Assume that \(q_j\) is not entangled with any other qubit in the current branch, so that its entanglement group is \(\{q_j\}\) and \(\gamma_Q(\{q_j\}) = s \neq \topq\), where \(s:\{0,1\}\rightharpoonup\mathbb{C}\) maps basis strings to amplitudes (\cref{subsec:prelim-qcp}).
If there exists \(b\in\{0,1\}\) such that \(|s(b)|=1\) and \((1-b)\notin\operatorname{dom}(s)\), then the measured qubit \(q_j\) is already in the basis state \(\ket{b}\).
Hence, the measurement is deterministic, introduces no branching, and leaves the quantum state unchanged.
It only updates the classical state by writing \(b\) to \(c_i\), yielding the updated branch \((\gamma_C[c_i\mapsto b],(\mathcal{G},\gamma_Q))\).

\paragraph{Non-deterministic measurement}
When the measurement is not deterministic, either because the tracked
state admits multiple outcomes or because the measured qubit is untracked,
we generate two successor branches, one for each outcome \(b\in\{0,1\}\).
In each successor branch, the classical abstraction is updated by setting $c_i$ to $b$.
Moreover, the measurement disentangles $q_j$ from the rest of its entanglement group.
Let $G\in\mathcal{G}$ be the group such that $q_j\in G$.
The partition is then refined by replacing $G$ with ${q_j}$ and $G\setminus{q_j}$, omitting the latter when it is empty.
The singleton group ${q_j}$ is set to represent the basis state $\ket{b}$.
If $\gamma_Q(G)\neq\topq$, the abstract state of the remaining qubits $G\setminus{q_j}$ is updated by conditioning on the outcome $b$: we keep only the amplitudes of the pre-measurement superposition in which $q_j=b$, discard the others, and renormalize the result.
If instead $\gamma_Q(G)=\topq$, no conditional post-measurement state can be computed for $G\setminus{q_j}$, which is therefore kept untracked.

\paragraph*{Transfer function definition}
Let \(\mathcal{B}\subseteq \mathcal{C}_C\times\mathcal{Q}_Q\) be the current abstract state and consider a measurement \(M_{q_j \to c_i}\).
For a branch \((\gamma_C,(\mathcal{G},\gamma_Q)) \in \mathcal{B}\), let
\(\mathsf{S}^{\sharp}(\gamma_C,(\mathcal{G},\gamma_Q))\) denote the set of successor branches obtained by applying \(M_{q_j \to c_i}\) according to the deterministic and non-deterministic cases defined above.
Hence, \(|\mathsf{S}^{\sharp}(\gamma_C,(\mathcal{G},\gamma_Q))|\in\{1,2\}\).
While a deterministic measurement does not increase the number of abstract branches, as the current branch is simply replaced by its updated successor, a non-deterministic measurement produces two successor branches and therefore increases the number of represented branches by one.
If such a split would cause the global bound \(b_{\max}\) to be exceeded, then the two successor branches are replaced by a single conservative fallback branch defined as
\[
\mathsf{FB}^{\sharp}(\gamma_C,(\mathcal{G},\gamma_Q))
\triangleq
\{(\gamma_C[c_i\mapsto \topc],\ (\mathcal{G},\gamma_Q[G\mapsto \topq]))\},
\]
where \(G\in\mathcal{G}\) is the entanglement group such that \(q_j\in G\).
Accordingly, the transfer function \(\atrans{M_{q_j \to c_i}}{\mathcal{B}}\) is defined by processing the branches of \(\mathcal{B}\) in any fixed order while respecting the global bound \(b_{\max}\).
Let \(r \triangleq b_{\max}-|\mathcal{B}|\) and initialize \(\mathcal{B}'\triangleq\emptyset\).
For each branch \((\gamma_C,(\mathcal{G},\gamma_Q))\in\mathcal{B}\), let
\(S\triangleq \mathsf{S}^{\sharp}(\gamma_C,(\mathcal{G},\gamma_Q))\).
If \(|S|=2\) and \(r=0\), we set \(S \gets \mathsf{FB}^{\sharp}(\gamma_C,(\mathcal{G},\gamma_Q))\), if \(|S|=2\) and \(r>0\), set \(r\gets r-1\).
The resulting set \(S\) is then added to \(\mathcal{B}'\), i.e., \(\mathcal{B}' \gets \mathcal{B}' \cup S\).
The transfer function for measurements is defined as
\[
\atrans{M_{q_j \to c_i}}{\mathcal{B}} \triangleq \mathcal{B}'.
\]

\subsubsection{Transfer function for resets}
\label{sec:method-transfer-reset}
Resets are non-unitary operations that reinitialize a qubit to the state $\ket{0}$.
In the concrete semantics, when the reset qubit is entangled, resetting it breaks the entanglement and the
remaining qubits entangled with the reset qubit are left in a mixed state that depends on the pre-reset state, while leaving the classical state unchanged.
Let $R_{q_j}$ denote the instruction that resets the qubit $q_j$.
For a branch $(\gamma_C,(\mathcal{G},\gamma_Q))$, let $G\in\mathcal{G}$ be the entanglement group such that
$q_j\in G$.
Let \(s_0 : \{0,1\} \rightharpoonup \mathbb{C}\) denote the abstract state such that \(s_0(0)=1\) and undefined elsewhere. If \(G=\{q_j\}\), then we set the abstract state of the singleton group \(\{q_j\}\) to \(s_0\), while leaving the qubit partition unchanged.
If $|G|>1$, the reset breaks the entanglement between $q_j$ and the other qubits in $G$.
We replace $G$ in the partition with two groups, $\{q_j\}$ and $G\setminus\{q_j\}$, obtaining a refined
partition $\mathcal{G}'$.
In our abstract semantics, the information
about $G\setminus\{q_j\}$ is conservatively discarded by setting $\gamma_Q'(G\setminus\{q_j\}) = \topq$, while $\gamma_Q'(\{q_j\})$ is set to $s_0$.
Let $\mathcal{B}$ be the current abstract state and consider a reset $R_{q_j}$. The transfer function for resets is defined as
\[
\begin{aligned}
  \atrans{R_{q_j}}{\mathcal{B}}
  \;\triangleq\;
  \{&\,(\gamma_C,(\mathcal{G}',\gamma_Q')) \mid\
      (\gamma_C,(\mathcal{G},\gamma_Q))\in\mathcal{B} \\
  & \wedge\ (\mathcal{G}',\gamma_Q')=\mathsf{RU}_{q_j}(\mathcal{G},\gamma_Q)\,\}.
\end{aligned}
\]
Here \(\mathsf{RU}_{q_j}\) updates the quantum component according to the reset semantics described above.

\subsubsection{Transfer function for if-then-else operations}
\label{sec:method-transfer-ifthenelse}
Conditional operations execute instruction blocks according to the evaluation of a classical condition $\varphi$ over the classical register $C$. Given a conditional instruction
$\mathbf{if}\ (\varphi)\ \mathbf{then}\ P_t\ \mathbf{else}\ P_f$
and an incoming abstract state
$\mathcal{B}\subseteq \mathcal{C}_C\times\mathcal{Q}_Q$,
our abstract semantics partitions $\mathcal{B}$ into the branches that may reach the \emph{then} block and the branches that may reach the \emph{else} block, based on the abstract guard evaluation
$\llbracket \varphi \rrbracket^{\sharp} : \mathcal{C}_C \to \mathbb{B}_3$
defined in \cref{sec:method-guards-eval}.
Accordingly, we define
\[
  \mathcal{B}_{\mathsf{t}}
  \;\triangleq\;
  \{\,(\gamma_C,\eta_Q)\in\mathcal{B}\ \mid\
     \llbracket \varphi \rrbracket^{\sharp}(\gamma_C)\in\{\mathsf{t},\mathsf{u}\}\,\},
\]
\[
  \mathcal{B}_{\mathsf{f}}
  \;\triangleq\;
  \{\,(\gamma_C,\eta_Q)\in\mathcal{B}\ \mid\
     \llbracket \varphi \rrbracket^{\sharp}(\gamma_C)\in\{\mathsf{f},\mathsf{u}\}\,\}.
\]
$\mathcal{B}_{\mathsf{t}}$ collects the branches in which the \emph{then} block may be taken, i.e., those for which the guard evaluates to $\mathsf{t}$ or $\mathsf{u}$. Dually, $\mathcal{B}_{\mathsf{f}}$ collects the branches in which the \emph{else} block may be taken, i.e., those for which the guard evaluates to $\mathsf{f}$ or $\mathsf{u}$.

\paragraph{Propagation through the two blocks}
After partitioning the incoming abstract state $\mathcal{B}$ into the two branch sets
$\mathcal{B}_{\mathsf{t}}$ and $\mathcal{B}_{\mathsf{f}}$, we analyze the two sides of the conditional
independently, propagating the \emph{then} block $P_t$ from $\mathcal{B}_{\mathsf{t}}$ and the \emph{else}
block $P_f$ from $\mathcal{B}_{\mathsf{f}}$.
Both $P_t$ and $P_f$ are instruction sequences, denoted by
$P_t = (I^t_1,\dots,I^t_{L_t})$ and $P_f = (I^f_1,\dots,I^f_{L_f})$. The abstract effect of each block is obtained by composing the transfer functions
of its instructions:
\begin{align*}
    \mathfrak{T}^{\sharp}(P_t)(\mathcal{B}_{\mathsf{t}})
    &\;\triangleq\;
    \mathfrak{T}^{\sharp}(I^t_{L_t})(
      \cdots \mathfrak{T}^{\sharp}(I^t_2)(
      \mathfrak{T}^{\sharp}(I^t_1)(\mathcal{B}_{\mathsf{t}})
      )\cdots
    ),
\\
    \mathfrak{T}^{\sharp}(P_f)(\mathcal{B}_{\mathsf{f}})
    &\;\triangleq\;
    \mathfrak{T}^{\sharp}(I^f_{L_f})(
      \cdots \mathfrak{T}^{\sharp}(I^f_2)(
      \mathfrak{T}^{\sharp}(I^f_1)(\mathcal{B}_{\mathsf{f}})
      )\cdots
    ).
\end{align*}
The \emph{then} and \emph{else} blocks are propagated independently. Since $\mathcal{B}_{\mathsf{t}},\mathcal{B}_{\mathsf{f}}\subseteq \mathcal{B}$ and $|\mathcal{B}|\le b_{\max}$, each block starts from at most $b_{\max}$ branches. Moreover, each instruction transfer function enforces the bound $n_{\max}$ branchwise, and the transfer function for measurements enforces the bound $b_{\max}$ as in \cref{sec:method-transfer-measure}. Therefore, every intermediate abstract state produced during the propagation of either block remains bounded by $b_{\max}$.

\paragraph{Abstract join after the conditional}
The abstract state at the join point of the conditional is obtained as the union of the two successor sets:
\[
  \mathcal{B}_{\mathsf{join}}
  \;\triangleq\;
  \atrans{P_t}{\mathcal{B}_{\mathsf{t}}} \ \cup\  \atrans{P_f}{\mathcal{B}_{\mathsf{f}}}.
\]
Even though each side is individually bounded by $b_{\max}$, their union may contain up to $2 \times b_{\max}$ branches.
To maintain the invariant $|\mathcal{B}|\le b_{\max}$, we apply a post-join branch-reduction operator
$\mathsf{RD}_{b_{\max}}$ defined as follows:
\[
  \mathsf{RD}_{b_{\max}}(\mathcal{B})
  \;\triangleq\;
  \begin{cases}
    \mathcal{B} & \text{if } |\mathcal{B}|\le b_{\max},\\
    \mathcal{B}' & \text{otherwise.}
  \end{cases}
\]
Here $\mathcal{B}'$ is obtained from $\mathcal{B}$ by repeatedly selecting two branches
$A,B\in\mathcal{B}$ in any fixed deterministic order and replacing them with their conservative merge
$A\sqcup B$ from \cref{sec:method-branch-merge}, until the cardinality drops to $b_{\max}$.
Since $A\sqcup B$ conservatively over-approximates the disjunction of $A$ and $B$, iterating this reduction
preserves soundness while trading precision.

\paragraph*{Transfer function definition}
The transfer function for a $I = \mathbf{if}(\varphi)\ \mathbf{then}\ P_t\ \mathbf{else}\ P_f$ instruction is
\[
\begin{aligned} \atrans{I}{\mathcal{B}} 
  \;\triangleq\; 
  \mathsf{RD}_{b_{\max}}
  (
    \mathcal{B}_\mathsf{join}
  ).
\end{aligned}
\]

\subsection{Circuit simplification rules using abstract states}
\label{sec:method-simpl}
We now show how the propagated information can be used to simplify dynamic circuits while preserving
their concrete classical--quantum semantics. These simplifications are driven by the branch-disjunctive
abstract state computed at each program point. In the presence of \textit{if-then-else}
constructs, the branch-aware propagation enables block-local simplifications inside the conditional. Moreover, the join point after the conditional retains a sound over-approximation of the information carried by
the incoming branches, so simplification can continue after the join as well.
\subsubsection{Control reduction under branching}
\label{sec:method-optim-control}
We revisit control reduction introduced in QCP (\cref{subsec:prelim-qcp-cont-red}) in the context of our branch-aware analysis.
In our setting, the abstract state is a finite set of execution branches, and control conditions are evaluated independently on each branch. To preserve soundness, a control can be dropped only if it is redundant in every branch of the current abstract state $\mathcal{B}$, and a controlled operation can be removed only if its controls are unsatisfiable in every branch.
\Cref{fig:exmp_toff_red} shows an example of how BQCP enables control reductions that QCP cannot derive.
\begin{figure}[t]
\centering
\subfloat[\label{fig:exmp_toff_red_orig}]{%
\resizebox{0.37\linewidth}{!}{%
\begin{quantikz}[row sep=0.3cm, column sep=0.5cm]
\lstick{$\ket{0}$} & \gate{H} & \ctrl{1} & \meter{c_0}\wire[d][3]{c} & \qw & \ctrl{2} & \qw \\
\lstick{$\ket{0}$} & \gate{X} & \targ{} & \qw &  \meter{c_1}\wire[d][2]{c} & \ctrl{1} & \qw \\
\lstick{} & \qw & \qw & \qw & \qw & \targ{} & \qw \\
\lstick{$c$} & \setwiretype{c} \cw & \cw & \cw & \cw & \cw & \cw
\end{quantikz}
}}
\qquad
\subfloat[\label{fig:exmp_toff_red_simp}]{%
\resizebox{0.34\linewidth}{!}{%
\begin{quantikz}[row sep=0.3cm, column sep=0.5cm]
\lstick{$\ket{0}$} & \gate{H} & \ctrl{1} & \meter{c_0}\wire[d][3]{c} & \qw & \qw \\
\lstick{$\ket{0}$} & \gate{X} & \targ{}  & \qw & \meter{c_1}\wire[d][2]{c} & \qw \\
\lstick{} & \qw & \qw & \qw & \qw & \qw \\
\lstick{$c$} & \setwiretype{c} \cw & \cw & \cw & \cw & \cw
\end{quantikz}
}}
\caption{Example of BQCP eliminating an unreachable Toffoli gate: \Cref{fig:exmp_toff_red_orig} shows the original circuit, while \Cref{fig:exmp_toff_red_simp} shows the simplified circuit. The first two qubits are prepared in the Bell state $(\ket{01}+\ket{10})/\sqrt{2}$, so after measurement they are never simultaneously equal to $1$. Hence, the Toffoli control condition is unsatisfiable in every execution branch.}
\label{fig:exmp_toff_red}
\end{figure}
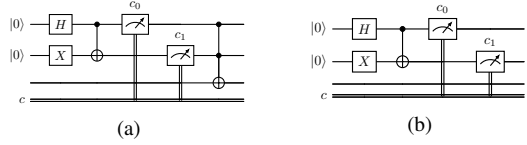
\subsubsection{Elimination of semantically redundant operations}
\label{sec:method-el-red-op}
We use the abstract state at each program point to detect instructions that have no effect on the current
classical--quantum state of the circuit. An instruction can be removed whenever the abstract state is precise enough to establish that it is semantically redundant.
\paragraph{Redundant unitaries}
Let $U$ be a unitary gate acting on the qubits $S\subseteq Q$. For each branch
$(\gamma_C,(\mathcal{G},\gamma_Q))\in\mathcal{B}$, let
$(\mathcal{G}',\gamma_Q') \triangleq \tau^\sharp_U((\mathcal{G},\gamma_Q))$, where
$\tau^\sharp_U : \mathcal{Q}_Q \to \mathcal{Q}_Q$ is the QCP transfer function
(\cref{sec:method-transfer}).
If the abstract state is precise enough to establish that $U$ is semantically redundant in every branch of $\mathcal{B}$, we can remove it. A sufficient condition is that, for every branch, all entanglement groups intersecting $S$ are tracked and left unchanged by $U$, i.e.,
\[
  \forall\,G\in\mathcal{G}, G\cap S\neq\emptyset.\;
  G\in\mathcal{G}' \wedge \gamma_Q(G)\neq\topq \wedge \gamma_Q'(G)=\gamma_Q(G).
\]
Intuitively, for each branch, the restriction of the partition to the qubits in $S$ is unchanged and the abstract state of every affected group is preserved. Hence, $U$ leaves the concrete quantum state unchanged in every execution represented by that branch. When this holds for all branches in $\mathcal{B}$, we say that $U$ is redundant under $\mathcal{B}$, and in that case, $U$ can be removed.

\paragraph{Redundant resets} A reset operation $R_{q_j}$ is redundant when the qubit $q_j$ is tracked as the deterministic basis state $\ket{0}$. Formally, for an abstract program state $\mathcal{B}$, we remove $R_{q_j}$ if, for every branch
$(\gamma_C,(\mathcal{G},\gamma_Q))\in\mathcal{B}$, the entanglement group containing $q_j$ is the singleton
$\{q_j\}$ and
\(
  \gamma_Q(\{q_j\}) = s_0,
\)
where $s_0$ is the abstract encoding of $\ket{0}$ (i.e., $s_0(0)=1$ and
undefined elsewhere). In this case, executing $R_{q_j}$ leaves the concrete quantum state unchanged.
\paragraph{Redundant measurements}
A measurement operation $M_{q_j \to c_i}$ is redundant under $\mathcal{B}$ when, in every execution represented by $\mathcal{B}$, it leaves the concrete classical--quantum state unchanged. Formally, given an abstract program state
$\mathcal{B}\subseteq \mathcal{C}_C\times\mathcal{Q}_Q$,
we remove $M_{q_j \to c_i}$ if, for every branch
$(\gamma_C,(\mathcal{G},\gamma_Q))\in\mathcal{B}$, the entanglement group containing $q_j$ is the singleton
$\{q_j\}$ and the state of $q_j$ is precisely tracked as a deterministic basis state $\ket{b}$ for some
$b\in\{0,1\}$, with the target classical bit already equal to $b$, i.e., $\gamma_C(c_i)=b$.
In this case, the measurement of $q_j$ does not change the concrete classical--quantum state and is therefore semantically redundant.

\subsubsection{Optimizations for conditionals}
\label{sec:method-optim-if}
Consider a conditional instruction
$\mathbf{if}\ (\varphi)\ \mathbf{then}\ P_t\ \mathbf{else}\ P_f$
analyzed under an abstract program state $\mathcal{B}$. Using the abstract guard
evaluation $\llbracket \varphi \rrbracket^{\sharp}$, we partition $\mathcal{B}$ into $\mathcal{B}_{\mathsf{t}}$ and $\mathcal{B}_{\mathsf{f}}$ (\cref{sec:method-transfer-ifthenelse}),
representing the branches that may reach the \emph{then} and \emph{else} blocks,
respectively. We exploit $\mathcal{B}_{\mathsf{t}}$ and $\mathcal{B}_{\mathsf{f}}$  to simplify both the conditional
and the instructions within its two blocks.

\paragraph{Eliminating deterministic conditionals}
If $\mathcal{B}_{\mathsf{f}}=\emptyset$, then the guard cannot evaluate to false on any branch compatible with
$\mathcal{B}$, and the conditional can be replaced by the \emph{then} block $P_t$. Symmetrically, if
$\mathcal{B}_{\mathsf{t}}=\emptyset$, only the \emph{else} block $P_f$ is reachable and the conditional can be
replaced by $P_f$. Eliminating such deterministic conditionals simplifies the dynamic circuit by removing
runtime control-flow decisions: the circuit no longer needs to evaluate $\varphi$ and select between $P_t$ and
$P_f$ during execution.
This not only simplifies the dynamic circuit itself, but also improves the effectiveness of subsequent compilation passes, which can then operate on unconditional blocks instead of reasoning about conditional control flow.
\paragraph{Block-local simplification}
When both $\mathcal{B}_{\mathsf{t}}$ and $\mathcal{B}_{\mathsf{f}}$ are non-empty, the \emph{then} and \emph{else} blocks are analyzed independently under the branch-restricted states $\mathcal{B}_{\mathsf{t}}$ and $\mathcal{B}_{\mathsf{f}}$, respectively. This preserves the correlation between the guard outcome and the corresponding
classical--quantum information, allowing each block to be simplified using only the branches that may reach it. In particular, we simplify $P_t$ using $\mathcal{B}_{\mathsf{t}}$, and we simplify $P_f$ using $\mathcal{B}_{\mathsf{f}}$. This enables optimizations that are local to one side of the conditional, such as control reduction or the elimination of semantically redundant operations.
\Cref{fig:exmp_contr_red} illustrates a block-local simplification enabled by our branch-aware analysis.
BQCP determines that the control qubit is in state $\ket{1}$ in the \emph{if} block and in state $\ket{0}$ in the \emph{else} block. Hence, the control on $U_1$ is redundant in the former, whereas the control on $U_2$ can never be triggered in the latter. QCP, which does not track multiple branches, conservatively retain both controlled operations.
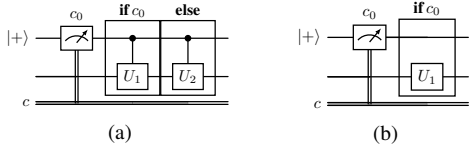
\begin{figure}[t]
\centering
\subfloat[\label{fig:exmp_contr_red_orig}]{%
\resizebox{0.355\linewidth}{!}{%
  \begin{quantikz}[row sep=0.3cm, column sep=0.55cm]
\lstick{$\ket{+}$} & \meter{c_0}\wire[d][2]{c} & \ctrl{1}\gategroup[wires=2,steps=1,
  style={draw,inner xsep=4pt,inner ysep=0.1pt},
  label style={label position=above}]{\textbf{if}$\,c_0$} & \ctrl{1}\gategroup[wires=2,steps=1,
  style={draw,inner xsep=4pt,inner ysep=0.1pt},
  label style={label position=above}]{\textbf{else}} &
\\
\lstick{} & \qw & \gate{U_1} & \gate{U_2} &
\\
\lstick{$c$}   & \setwiretype{c} \cw & \cw & \cw &
\end{quantikz}
}}
\qquad
\subfloat[\label{fig:exmp_contr_red_semp}]{%
\resizebox{0.28\linewidth}{!}{%
  \begin{quantikz}[row sep=0.3cm, column sep=0.55cm]
\lstick{$\ket{+}$} & \meter{c_0}\wire[d][2]{c} & \qw \gategroup[wires=2,steps=1,
  style={draw,inner xsep=4pt,inner ysep=0.1pt},
  label style={label position=above}]{\textbf{if}$\,c_0$} &
\\
\lstick{} & \qw & \gate{U_1} &
\\
\lstick{$c$}   & \setwiretype{c} \cw & \cw &
\end{quantikz}
}}
\caption{Example of a block-local simplification with BQCP: \Cref{fig:exmp_contr_red_orig} shows the original circuit, while \Cref{fig:exmp_contr_red_semp} shows the corresponding simplified circuit.}
\label{fig:exmp_contr_red}
\end{figure}

\section{Correctness and Complexity}
\label{sec:correctness-complexity}
In this section, we discuss the correctness and complexity of the proposed method.
\subsection{Correctness}
\label{sec:correctness-repr}
We establish the soundness of our abstract state propagation with respect to the concrete classical--quantum semantics defined in \Cref{subsec:pre-qc-sem}. We also show that our simplification rules preserve circuit semantics.
\subsubsection{Correctness of state propagation}
We define a representation relation $\models$ between concrete classical--quantum states $\mathfrak{S}$ and abstract program states $\mathcal{B}$, written $\mathfrak{S} \models \mathcal{B}$, to express that $\mathcal{B}$ soundly over-approximates all concrete execution branches in $\mathfrak{S}$.
The soundness of abstract propagation is proved by induction on program instructions. Specifically, for every instruction $I$, if $\mathfrak{S} \models \mathcal{B}$ holds at a program point, then the abstract transfer function for $I$ soundly over-approximates the concrete semantics of $I$, i.e.,
\(
  \ctrans{I}{\mathfrak{S}} \models \atrans{I}{\mathcal{B}}.
\)
Here $\ctrans{I}{\mathfrak{S}}$ denotes the concrete transition induced by $I$ on the concrete state $\mathfrak{S}$, while $\atrans{I}{\mathcal{B}}$ denotes the corresponding abstract transfer function defined in \Cref{sec:method-transfer}.


For an abstract classical state \(\gamma_C\in\mathcal{C}_C\), we define the set of compatible
concrete classical configurations
\[
  \Gamma_C(\gamma_C)
  \triangleq
  \{\, c \in \{0,1\}^m \mid
     \forall i.\ \gamma_C(c_i) \neq \topc \Rightarrow c_i = \gamma_C(c_i)
  \,\}.
\]

Given $\eta_Q=(\mathcal{G},\gamma_Q)\in\mathcal{Q}_Q$. For each tracked group
$G\in\mathcal{G}$ with $\gamma_Q(G)=s\neq \topq$, let $\ket{\psi_s}$ be the normalized state encoded by $s$
and $\rho_s \triangleq \ket{\psi_s}\!\bra{\psi_s}$. The set of compatible concrete quantum configurations is
\begin{align*}
\Gamma_Q(\eta_Q)
\triangleq
\{ \rho \in \mathcal{D}(\mathcal{H}_Q) \,|\, \forall G\in\mathcal{G}.\gamma_Q(G)=s \Rightarrow
\mathfrak{t}_G(\rho)=\rho_s \}.
\end{align*}
Here
$\mathfrak{t}_G(\rho) \triangleq \Tr_{Q\setminus G}(\rho)$ is the partial trace of $\rho$ over the complement of $G$. In other words, whenever $\gamma_Q(G)$ is tracked (i.e., $\gamma_Q(G)=s\neq \topq$), the reduced state of $\rho$
on $G$ must coincide with $\rho_s$.
Since $\rho_s$ is pure, $\mathfrak{t}_G(\rho)=\rho_s$ implies that $G$ is uncorrelated with $Q\setminus G$.

A single abstract branch $(\gamma_C,\eta_Q)\in \mathcal{C}_C\times\mathcal{Q}_Q$ represents the set of
compatible concrete configurations
\[
  \Gamma_{CQ}(\gamma_C,\eta_Q)
  \triangleq
  \{\, (c,\rho)\mid c\in \Gamma_C(\gamma_C)\ \wedge\ \rho\in \Gamma_Q(\eta_Q)\,\}.
\]
The set of concrete classical-quantum states represented by an abstract program state $\mathcal{B}\subseteq\mathcal{C}_C\times \mathcal{Q}_Q$ is
\[
  \Gamma(\mathcal{B})
  \triangleq
  \bigcup_{(\gamma_C,\eta_Q)\in\mathcal{B}} \Gamma_{CQ}(\gamma_C,\eta_Q).
\]
$\mathcal{B}$ represents $\mathfrak{S}$, written $\mathfrak{S}\models \mathcal{B}$, if
\(
  \mathfrak{S} \subseteq \Gamma(\mathcal{B}).
\)
When $\mathcal{B}=\emptyset$, we have $\Gamma(\mathcal{B})=\emptyset$, so $\mathfrak{S}\models\mathcal{B}$
holds if and only if $\mathfrak{S}=\emptyset$.


We now state some lemmas that will be used to prove the soundness of propagation.
\begin{lemma}[Soundness of abstract guard evaluation]
\label{lem:guard-sound}
Let $\llbracket \varphi \rrbracket(c)\in\{t,f\}$ denote the concrete Boolean evaluation of $\varphi$ on $c$.
For every guard $\varphi$, abstract state $\gamma_C\in\mathcal{C}_C$, and concrete configuration
$c\in\Gamma_C(\gamma_C)$, the following implications hold: \(
\llbracket \varphi \rrbracket^{\sharp}(\gamma_C)=\mathsf{t} \Rightarrow \llbracket \varphi \rrbracket(c)=t\), \(\
\llbracket \varphi \rrbracket^{\sharp}(\gamma_C)=\mathsf{f} \Rightarrow \llbracket \varphi \rrbracket(c)=f.
\)
\end{lemma}
\begin{proof}
By structural induction on $\varphi$, using the definition of $\Gamma_C$ and $\llbracket \cdot \rrbracket^{\sharp}$.
\end{proof}

\begin{lemma}[Soundness of branch merge]
\label{lem:branch-merge-sound}
Let $A=(\gamma_C^A,\eta_Q^A)$ and $B=(\gamma_C^B,\eta_Q^B)$ be two abstract branches, and let
$A\sqcup B=(\gamma_C^{A\sqcup B},\eta_Q^{A\sqcup B})$ be their conservative merge. Then
\[
\Gamma_{CQ}(\gamma_C^A,\eta_Q^A)\ \cup\ \Gamma_{CQ}(\gamma_C^B,\eta_Q^B)
\ \subseteq\
\Gamma_{CQ}(\gamma_C^{A\sqcup B},\eta_Q^{A\sqcup B}).
\]
\end{lemma}
\begin{proof}
\emph{Classical part}: by definition of $\gamma_C^{A\sqcup B}$, a bit is kept constant only when both branches agree,
and it is set to $\topc$ otherwise. Hence
$\Gamma_C(\gamma_C^A)\subseteq \Gamma_C(\gamma_C^{A\sqcup B})$ and similarly for $B$.
\emph{Quantum part:} by construction, $\eta_Q^{A\sqcup B}$ preserves a constraint on a group $G$ only when
$G$ occurs identically in both $A$ and $B$, and $\gamma_Q^A(G)=\gamma_Q^B(G)=s\neq \topq$. In this case,
$\gamma_Q^{A\sqcup B}(G)=s$, so the same constraint $\mathfrak{t}_G(\rho)=\rho_s$ is preserved in the merge.
All remaining qubits are placed in a (possibly empty) additional group whose abstract state is $\topq$, and therefore no constraint is imposed on them. Thus, every constraint present in $\eta_Q^{A\sqcup B}$ is already satisfied by any $\rho\in\Gamma_Q(\eta_Q^A)$, and similarly by any $\rho\in\Gamma_Q(\eta_Q^B)$. Hence
$\Gamma_Q(\eta_Q^A)\subseteq \Gamma_Q(\eta_Q^{A\sqcup B})$ and
$\Gamma_Q(\eta_Q^B)\subseteq \Gamma_Q(\eta_Q^{A\sqcup B})$.
\end{proof}

\begin{lemma}[Soundness of branch reduction]
\label{lem:reduce-sound}
For any abstract state $\mathcal{B}$, letting $\mathsf{RD}_{b_{\max}}(\mathcal{B})=\mathcal{B}'$, we have
\(
  \Gamma(\mathcal{B}) \subseteq \Gamma(\mathcal{B}').
\)
\end{lemma}
\begin{proof}
The operator $\mathsf{RD}_{b_{\max}}$ is defined by repeated replacement of two branches with their conservative
merge. The claim then follows by repeated application of \Cref{lem:branch-merge-sound}.
\end{proof}

\begin{lemma}[Soundness of instruction transfer functions]
\label{lem:instr-sound}
Let $I$ be any instruction. If $\mathfrak{S}\models\mathcal{B}$, then
\(
  \ctrans{I}{\mathfrak{S}} \models \atrans{I}{\mathcal{B}}.
\)
\end{lemma}
\begin{proof}
By cases on $I$:
\begin{itemize}
  \item \emph{Unitary $U$.}
  The transfer function applies the QCP transformer $\tau_U^\sharp$ to each branch.
  Since $U$ does not write the classical register, each branch preserves $\gamma_C$.
  Soundness of the quantum update follows directly from the soundness of $\tau_U^\sharp$~\cite{qcp}, which
  over-approximates $\rho\mapsto U\rho U^\dagger$ on tracked groups and conservatively maps
  groups to $\topq$ whenever information is not tracked or the $n_{\max}$ bound is exceeded. Hence, every
  concrete successor $(c,U\rho U^\dagger)$ belongs to $\Gamma(\atrans{I}{\mathcal{B}})$.

  \item \emph{Measurement $M_{q_j\to c_i}$.}
  For each branch, the abstract transfer function applies the measurement semantics locally. If the measurement is deterministic, i.e., the measured qubit is already in a basis state $\ket{b}$, the transfer function tracks $c_i\gets b$, and leaves the quantum state unchanged. Otherwise, the measurement is non-deterministic. If $b_{\max}$ is not exceeded, the transfer function creates the two successor branches corresponding to $b\in\{0,1\}$. In each branch, the classical bit is set to $c_i=b$. If the group containing $q_j$ is tracked, its abstract state is refined to reflect the outcome $q_j=b$; otherwise, the measured qubit is recorded as $\ket{b}$ while the remaining qubits in the group remain untracked. If splitting would exceed $b_{\max}$, the fallback branch sets $c_i\mapsto\topc$ and marks the group containing $q_j$ as untracked ($\topq$), thereby forgetting the outcome and any post-measurement constraint on that group. Hence, $\ctrans{M_{q_j\to c_i}}{\mathfrak{S}} \models \atrans{M_{q_j\to c_i}}{\mathcal{B}}$.

  \item \emph{Reset $R_{q_j}$.}
  For each branch, we apply $\mathsf{RU}_{q_j}$. Let
  $G$ be the entanglement group containing $q_j$. If $G=\{q_j\}$ we overwrite its state to represent $\ket{0}$; if $|G|>1$
  we split $G$ into $\{q_j\}$ and $G\setminus\{q_j\}$, set $\{q_j\}$ to represent $\ket0$, and set $G\setminus\{q_j\}$ to
  $\topq$, thus dropping any constraint on the remaining qubits. In both cases, the resulting abstract quantum state soundly
over-approximates the corresponding concrete post-reset states. Since a reset does not alter $\gamma_C$, we conclude 
  $\ctrans{R_{q_j}}{\mathfrak{S}} \models \atrans{R_{q_j}}{\mathcal{B}}$.

  \item \emph{Conditional $\mathbf{if}(\varphi)\ \mathbf{then}\ P_t\ \mathbf{else}\ P_f$.}
  Consider a concrete branch $(c,\rho)$ represented by some $(\gamma_C,\eta_Q)\in\mathcal{B}$, so
  $c\in\Gamma_C(\gamma_C)$. If $\llbracket\varphi\rrbracket(c)=t$, then by
  Lemma~\ref{lem:guard-sound} we cannot have $\llbracket\varphi\rrbracket^\sharp(\gamma_C)=\mathsf{f}$, hence
  $(\gamma_C,\eta_Q)\in\mathcal{B}_{\mathsf{t}}$; symmetrically, if $\llbracket\varphi\rrbracket(c)=f$
  then $(\gamma_C,\eta_Q)\in\mathcal{B}_{\mathsf{f}}$. Thus, each concrete branch belongs to the abstract branch set corresponding to the outcome of the guard.
  Soundness for $P_t$ and $P_f$ follows by induction over their instruction sequences. Finally, the concrete
  semantics at the join is the union of the two continuations, and the abstract join takes the union and then
  applies $\mathsf{RD}$; soundness of $\mathsf{RD}$ follows from Lemma~\ref{lem:reduce-sound}.
\end{itemize}
\end{proof}

\begin{theorem}[Soundness of state propagation]
\label{thm:prop-sound}
Let $\mathcal{B}_0$ be the initial abstract state and $\mathfrak{S}_0$ the initial concrete state, with
$\mathfrak{S}_0\models \mathcal{B}_0$. For any instruction sequence (program) $P$,
\[
  \ctrans{P}{\mathfrak{S}_0}\ \models\ \atrans{P}{\mathcal{B}_0}.
\]
\end{theorem}

\begin{proof}
By induction on the length of $P$. For the base case $P=\epsilon$, we have
$\ctrans{\epsilon}{\mathfrak{S}_0}=\mathfrak{S}_0$ and $\atrans{\epsilon}{\mathcal{B}_0}=\mathcal{B}_0$, hence
$\ctrans{\epsilon}{\mathfrak{S}_0}\models \atrans{\epsilon}{\mathcal{B}_0}$.
For the inductive step, write $P~=~P';I$ and assume
$\ctrans{P'}{\mathfrak{S}_0}\models \atrans{P'}{\mathcal{B}_0}$.
By definition of sequence semantics,
$\ctrans{P}{\mathfrak{S}_0}=\ctrans{I}{\ctrans{P'}{\mathfrak{S}_0}}$ and
$\atrans{P}{\mathcal{B}_0}=\atrans{I}{\atrans{P'}{\mathcal{B}_0}}$.
Applying Lemma~\ref{lem:instr-sound} to $I$ we have
\(
  \ctrans{I}{\ctrans{P'}{\mathfrak{S}_0}}
  \models
  \atrans{I}{\atrans{P'}{\mathcal{B}_0}},
\)
which is exactly $\ctrans{P}{\mathfrak{S}_0}\models \atrans{P}{\mathcal{B}_0}$.
\end{proof}

 

\subsubsection{Correctness of simplification rules}
\label{sec:correctness-simpl}

We show that the circuit simplifications of
\Cref{sec:method-optim-control,sec:method-el-red-op,sec:method-optim-if}
preserve the concrete semantics. A rewrite $P\rightarrow P'$ is \emph{sound under} $\mathcal{B}$ if for
every concrete classical--quantum state $\mathfrak{S}$ such that $\mathfrak{S}\models \mathcal{B}$ we have
\(
  \ctrans{P}{\mathfrak{S}} \;=\; \ctrans{P'}{\mathfrak{S}}.
\)

\begin{proof}[Proof sketch]
Fix a program point with abstract state $\mathcal{B}$, and let $\mathfrak{S}$ be any concrete state such
that $\mathfrak{S}\models \mathcal{B}$. We prove $\ctrans{P}{\mathfrak{S}}=\ctrans{P'}{\mathfrak{S}}$ by cases
on the applied rewrite rule.

\begin{itemize}
  \item \emph{Control reduction.}
  The rule is applied only if the corresponding control condition is redundant or unsatisfiable in every
  branch of $\mathcal{B}$. Since $\mathfrak{S}\models\mathcal{B}$, every concrete branch is covered by some
  abstract branch in $\mathcal{B}$ and therefore satisfies the same control fact. Hence, the rewriting does not change the concrete semantics.

  \item \emph{Elimination of semantically redundant operations.}
  A unitary gate, reset, or measurement is removed only when $\mathcal{B}$ is sufficient to show that the
  instruction has no effect on any concrete branch represented by $\mathcal{B}$, i.e.,
  for all $(c,\rho)\in\Gamma(\mathcal{B})$ executing the instruction leaves $(c,\rho)$ unchanged. Therefore, removing the operation preserves the concrete successor set.

  \item \emph{Simplification of $\mathbf{if}(\varphi)\ \mathbf{then}\ P_t\ \mathbf{else}\ P_f$.}
  If $\mathcal{B}_{\mathsf{f}}=\emptyset$ (resp.\ $\mathcal{B}_{\mathsf{t}}=\emptyset$), then
  $\llbracket\varphi\rrbracket^\sharp(\gamma_C)=\mathsf{t}$ (resp.\ $\mathsf{f}$) for every branch in
  $\mathcal{B}$, and by Lemma~\ref{lem:guard-sound} the \emph{els}e (resp.\ \emph{then}) block is unreachable on all
  concrete executions compatible with $\mathcal{B}$. Hence, the conditional is equivalent to $P_t$
  (resp.\ $P_f$).
\end{itemize}

In all cases, the rewrite preserves the concrete successor set, hence $\ctrans{P}{\mathfrak{S}}=\ctrans{P'}{\mathfrak{S}}$.
\end{proof}

\subsection{Complexity analysis}
Consider a dynamic circuit over $n$ qubits and $m$ classical bits. Let $g$ denote the number of instructions and $J$ the number of conditionals. We first recall the complexity of QCP, which serves as the baseline for our analysis \cite{qcp}.
QCP bounds the number of tracked basis states per entanglement group by a fixed constant $n_{\max}$. As a result, updating the abstract quantum state for a single gate takes $O(n)$ time,
and the overall algorithm over $g$ instructions runs in $O(g\cdot n)$ time.
In terms of space, QCP stores at most $n_{\max}$ amplitudes for each tracked group and at most $n$ groups overall, giving a space complexity of $O(n)$, assuming that $n_{\max}$ is a constant.
In BQCP, the same QCP update is applied independently to each branch of an abstract state
$\mathcal{B}$ with $|\mathcal{B}|\le b_{\max}$. Hence, omitting the constant $n_{\max}$, state propagation requires $O(b_{\max}\cdot g\cdot n)$ time.
In addition, manipulating the classical state incurs an
additional cost of $O(m)$ per branch. Let $\ell$ denote the
maximum syntactic size of a guard. Since guards are evaluated
compositionally on their syntax tree, guard evaluation costs
$O(\ell)$ per branch. This results in an additional
$O(b_{\max}\cdot g\cdot m + b_{\max}\cdot J\cdot \ell)$ term.
The applicability of our simplification rules is checked independently on each branch using only the abstract information computed during propagation; therefore, these checks do not affect the overall asymptotic complexity.
Finally, each join applies $\mathsf{RD}_{b_{\max}}$, which performs at most $b_{\max}$ conservative merges. Each merge combines $m$ classical bits and at most $n$ groups of fixed size $n_{\max}$, and therefore costs $O(m+n)$ time. Hence, the total overhead due to joins is $O(J\cdot b_{\max}\cdot (m+n))$.
Overall, the running time is
\(O(b_{\max}\cdot g\cdot n
+ b_{\max}\cdot g\cdot m
+ b_{\max}\cdot J\cdot \ell
+ b_{\max} \cdot J \cdot (n+m)).\)
Since $J \le g$, this simplifies to
\(O(b_{\max}\cdot g\cdot (n+m+\ell)).\)
The abstract state representation requires
$O(b_{\max}\cdot (n\cdot n_{\max}+m))$ space. Under the assumption that $b_{\max}$ and $n_{\max}$ are fixed constants, this yields an overall complexity of $O(g\cdot (n+m+\ell))$ time and $O(n+m)$ space.

\section{Evaluation}
We now evaluate the effectiveness of BQCP on optimizing dynamic circuits. We consider two sets of experiments: one on a real-world use case and one on random circuits. We assess the impact of BQCP on circuit simplification and analyze its execution time.
\subsection{Circuit cutting case study}
\label{subsec:eval-cutting}
Circuit cutting is a technique for decomposing a large quantum circuit into smaller subcircuits that can be executed independently and whose results are later combined to reconstruct the original computation. 
Two main circuit cutting paradigms are commonly considered in the literature: wire cutting \cite{wire-cutting} and gate cutting \cite{gate_cutting_alg_qdislib}. The gate cutting algorithm proposed in~\cite{gate_cutting_alg_qdislib} generates subcircuits that rely on mid-circuit measurements and conditionals to implement the cut procedure. In practice, circuit cutting frameworks apply general rules that correctly implement the cutting procedure for a broad class of circuits, but do not always produce locally optimized subcircuits. In this case study, we use BQCP to optimize the dynamic subcircuits generated by the cutting procedure, by removing redundant operations at compile time, including mid-circuit measurements and conditionals.
\begin{table}
\caption{Circuit cutting case study: number of operations across all subcircuits, grouped by operation type, reported as Pass \cite{ibmHoareOptimizerlatest} / QCP $n_{\max}=512$ / BQCP $n_{\max}=512$, $b_{\max}=4$.}
\label{tab:subcircuits_qiskit_vs_cp8}
\centering
\scriptsize
\begin{tabular}{l|cccc}
\toprule
circuit & 1-qubit gates & 2-qubit gates & measurements & conditionals \\
\midrule
GHZ-4 & 164 / 228 / 116 & 12 / 12 / 0 & 48 / 48 / 16 & 48 / 48 / 12 \\
QAOA-6 & 840 / 840 / 756 & 360 / 360 / 192 & 48 / 48 / 48 & 48 / 48 / 40 \\
BV-7 & 744 / 768 / 688 & 180 / 180 / 160 & 48 / 48 / 44 & 48 / 48 / 24 \\
\bottomrule
\end{tabular}
\end{table}

\subsubsection{Experimental setup}
To perform circuit cutting, we use the \texttt{Qdislib} library~\cite{Qdislib}, which implements the gate-cutting algorithm proposed in~\cite{gate_cutting_alg_qdislib} and generates subcircuits containing measurements and conditionals.
We consider three representative input circuits drawn from the literature:
\begin{itemize}
    \item GHZ-4, a $4$-qubit circuit for preparing a GHZ state~\cite{Nielsen_Chuang_2010};
    \item QAOA-6, a $6$-qubit QAOA instance with depth $p=1$ on a ring graph, with an initial layer of Hadamard gates, followed by nearest-neighbor $CX-RZ(\gamma)-CX$ interactions along the ring, and a final layer of $RX(\beta)$ rotations, with $\gamma=0.7$ and $\beta=0.4$~\cite{qaoa};
    \item BV-7, a Bernstein--Vazirani circuit on $7$ data qubits and one ancilla qubit, with an all-ones secret string~\cite{bv}.
\end{itemize}
For each input circuit, we apply two gate-cut operations. The resulting number of subcircuits grows exponentially with the number of cuts, scaling as $O(6^k)$ for $k$ cuts~\cite{gate_cutting_alg_qdislib}.

\subsubsection{Results}
We apply BQCP to each subcircuit obtained from the cutting procedure and compare its optimization impact against the Qiskit \texttt{HoareOptimizer} pass~\cite{ibmHoareOptimizerlatest}, based on the compile-time optimization method of~\cite{assertion-based}, and against QCP. \Cref{tab:subcircuits_qiskit_vs_cp8} reports, for each input circuit, the number of operations aggregated over all resulting subcircuits. BQCP achieves larger gate reductions overall, and it is the only method that also removes measurements and conditionals.

\subsection{Evaluation on random circuits}
\label{subsec:eval-random}
To evaluate our method on a broader range of dynamic circuit patterns, we used a synthetic dataset of random circuits.
\subsubsection{Dataset generation}
To generate our dataset, we extended Qiskit's random circuit generator \cite{ibmRandom}.
In the original version, the generated dynamic components follow a fixed pattern: all qubits are measured together, and each conditional is a \emph{then} block of one gate. Our extended version produces more varied patterns by measuring a random subset of qubits, generating instruction blocks of variable size, and randomly inserting reset operations.
In our dataset, each conditional has a \textit{then} block of up to $10$ operations and an optional \textit{else} block of up to $10$ operations. We generated circuits of different sizes to evaluate how the pass scales with circuit size.
Each circuit is parameterized by a size parameter $\mathit{dim}\in\{1,\dots,8\}$: it consists of $10\times \mathit{dim}$ qubits and has depth $50\times \mathit{dim}$, excluding measurements and the instructions inside conditional blocks. For each value of $\mathit{dim}$, we generated $10$ independent circuits.

\subsubsection{Optimization impact}
\begin{figure}[t]
  \centering
    \resizebox{0.7\linewidth}{!}{\input{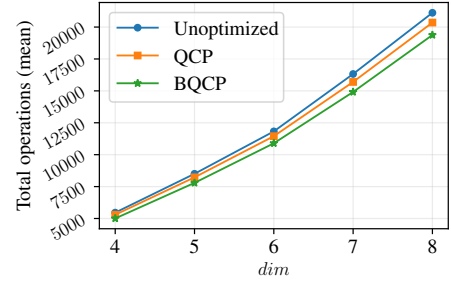}}
  \caption{Random circuits: Comparison of the mean number of operations (unitary, measurement, reset, and conditional) for the original unoptimized circuits and the circuits obtained after applying QCP (with $n_{\max}=512$) and BQCP (with $n_{\max}=512$, $b_{\max}=8$), as a function of the size parameter $\mathit{dim}$. Results are averaged over $10$ circuits for each $\mathit{dim}$.}
  \label{fig:tot_ops}
\end{figure}
\begin{table}[t]
\caption{Random circuits: mean multi-qubit gates reported as QCP $n_{\max}=512$ / BQCP $n_{\max}=512$, $b_{\max}=4$.}
\centering
\scriptsize
\begin{tabular}{l|ccc}
\toprule
$dim$ & 2-qubit gates & 3-qubit gates & 4-qubit gates \\
\midrule
2 & 336.7 / 321.0 & 206.9 / 188.6 & 91.9 / 79.5 \\
4 & 1254.3 / 1192.9 & 788.6 / 732.6 & 345.7 / 311.8 \\
6 & 2765.0 / 2662.3 & 1706.9 / 1608.2 & 794.1 / 733.4 \\
8 & 4799.5 / 4619.7 & 2956.5 / 2786.5 & 1408.3 / 1293.2 \\
\bottomrule
\end{tabular}
\label{tab:ops_by_qubits_qcp_vs_bqcp}
\end{table}

We evaluate the effectiveness of our approach by comparing the number of operations in the original unoptimized circuits with those obtained after applying QCP and BQCP independently. \Cref{fig:tot_ops} shows the mean number of operations as a function of the circuit size parameter $\mathit{dim}$. This count includes unitary gates, measurements, resets, and conditionals.
Each conditional counts as one control-flow operation, in addition to the operations inside its \textit{then} and \textit{else} blocks, which are counted separately.
We observe that BQCP consistently outperforms QCP across all circuit sizes, leading to a larger reduction in the number of operations. This is because BQCP analysis supports not only unitary gates, but also dynamic-circuit operations such as measurements, resets, and conditionals. As a result, it can propagate more information and identify more opportunities for simplification. Moreover, by tracking different execution branches, BQCP can also simplify operations inside conditional blocks.
\Cref{tab:ops_by_qubits_qcp_vs_bqcp} shows the reduction in the number of multi-qubit gates achieved by BQCP compared to QCP. The table reports the mean number of operations in the circuits optimized by QCP and BQCP, grouped by gate arity. Across all circuit sizes, BQCP consistently produces circuits with fewer multi-qubit gates than QCP, indicating a larger reduction in multi-qubit interactions within the circuit.

\subsubsection{Runtime analysis}
We analyze the execution time of BQCP as a function of the parameter $b_{\max}$. \Cref{fig:ex_time} reports both the mean execution time and its standard deviation across circuit sizes.
The execution time increases with both the circuit size and the value of $b_{\max}$, reflecting the cost of propagating a larger number of execution branches.
We observed in practice that increasing $b_{\max}$ can enable additional simplifications, but the improvements in the number of removed operations are not large.
This is likely because random circuits quickly generate highly entangled states. Therefore, even if the analysis tracks information, superpositions of states lead to multiple possible values for each qubit, limiting
the opportunities for simplifications.
As a result, larger values of $b_{\max}$ mainly increase runtime while providing only incremental reductions. These results suggest that moderate $b_{\max}$ values offer a practical balance: they retain most of the optimization benefits enabled by branch-aware propagation while avoiding the higher runtime costs observed for large values of $b_{\max}$ on larger circuits.
\begin{figure}[t]
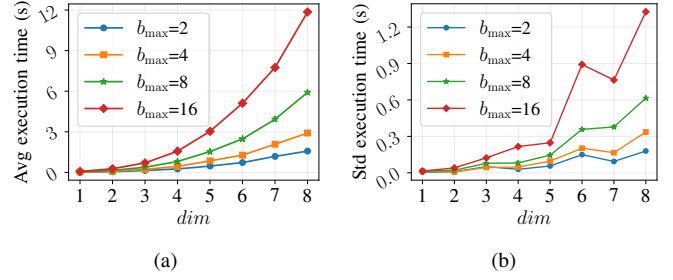

  \centering

  \subfloat[]{%
    \resizebox{0.495\linewidth}{!}{\input{images/avg_exec_time_mb_2-4-8-16.pgf}}%
    \label{subfig:ex_time_avg}%
  }\hfill
  \subfloat[]{%
    \resizebox{0.495\linewidth}{!}{\input{images/std_exec_time_mb_2-4-8-16.pgf}}%
    \label{subfig:ex_time_std}%
  }

  \caption{Random circuits: Mean execution time and standard deviation of BQCP with $n_{\max}=512$ for different values of $b_{\max}$, over $10$ circuits for each $\mathit{dim}$. \Cref{subfig:ex_time_avg} shows the mean execution time, and \Cref{subfig:ex_time_std} shows the standard deviation.}
  \label{fig:ex_time}
\end{figure}

\section{Conclusions and Future Works}
In this work, we introduced Branch-Aware Quantum Constant Propagation (BQCP), a compile-time optimization pass for dynamic quantum circuits. BQCP extends Quantum Constant Propagation (QCP) by tracking classical–quantum information across different execution branches, allowing to identify simplifications that are not captured by existing approaches not designed for dynamic circuits.
To ensure scalability, we proposed a bounded analysis that limits both the size of the tracked quantum state and the number of execution branches, providing a tunable trade-off between precision and cost. We formally proved the soundness of the analysis and of the simplifications.
Through experimental results we showed that BQCP achieves larger circuit reductions than existing optimization passes and standard QCP on dynamic circuits.

Several directions for future works remain open. 
First, the analysis could be extended to support richer dynamic control-flow constructs like loops or function calls.
Second, a possible direction is to investigate whether more refined strategies for managing the tracked branches could improve the precision of the analysis, for example through heuristics for merging branches that limit the loss of information.
Finally, it would be interesting to evaluate how integrating BQCP into compilation pipelines for quantum circuits affects the quality of the compiled circuits, analyzing how the circuits may become more efficient and more reliable when executed on real quantum hardware.

\section*{Acknowledgments}
The research is part of the Munich Quantum Valley (MQV), which is supported by the Bavarian state government with funds from the Hightech Agenda Bayern Plus.

We are grateful to Prof. Dr. Helmut Seidl for many fruitful discussions and his support at all times.

We used Codex (OpenAI) in the research artifacts to assist in the development of scripts for launching experimental runs and collecting experimental results. All AI-generated code were reviewed and validated by the authors.
\bibliographystyle{ieeetr}
\bibliography{references}

\end{document}